\documentclass[lettersize,onecolumn]{IEEEtran}
\usepackage{amsmath,amsfonts}
\usepackage{algorithmic}
\usepackage{algorithm}
\usepackage{array}
\usepackage[caption=false,font=normalsize,labelfont=sf,textfont=sf]{subfig}
\usepackage{textcomp}
\usepackage{stfloats}
\usepackage{url}
\usepackage{verbatim}
\usepackage{graphicx}
\usepackage{cite}
\usepackage{multirow}
\usepackage{bbm} 					
\usepackage{booktabs} 				
\usepackage{color}
\usepackage{amssymb}				

\usepackage{enumitem}				
\usepackage{titlesec}
\usepackage{amsthm} 				
\usepackage{hyperref}
\usepackage{float}
\usepackage{makecell}
\usepackage{tablefootnote}
\usepackage{xcolor} 
\usepackage[inner=1in,outer=1in,top=1in]{geometry}	
\usepackage{amsmath,amssymb,amsfonts}	
\usepackage[T1]{fontenc}
\usepackage[utf8]{inputenc}
\usepackage{authblk}

\floatstyle{boxed}

\hypersetup{
	colorlinks=true,    
	linkcolor=blue,     
	citecolor=blue,     
	urlcolor=blue       
}
\newtheorem{Thm}{Theorem}[section] 		
\newtheorem{Lem}[Thm]{Lemma}	
\newtheorem{Cor}[Thm]{Corollary}	
\newtheorem{Def}[Thm]{Definition}	
\newtheorem{Rmk}[Thm]{Remark}		

\newcommand{\ZZ}{\mathbb{Z}}
\newcommand{\FF}{\mathbb{F}}%

\newcommand{\EP}{\widetilde{\mathsf{EP}}_{\mathsf{RMFE}}}

\begin{document}

    \title{Coded Distributed (Batch) Matrix Multiplication over Galois Ring via RMFE}
    \date{}
    \author[*]{Yi Kuang}
    \author[*]{Jiang Li}
    \author[*]{Songsong Li}
    \author[*]{Chaoping Xing}
    \affil[*]{School of Electronic Information and Electrical Engineering, Shanghai Jiao Tong University, Shanghai,  China,\ \authorcr{schemer@sjtu.edu.cn, lijiang22222@sjtu.edu.cn, songsli@sjtu.edu.cn, xingcp@sjtu.edu.cn }}




    \maketitle

    \begin{abstract}
    Coded Distributed Matrix Multiplication (CDMM) is a distributed matrix multiplication (DMM) for large-scale matrices through a coding scheme such that any $R$ worker node among all $N$ worker nodes can recover the final product, where $N$ corresponds to the length of the code and $R\leq N$ is called the recovery threshold. The state-of-art CDMM schemes, such as EP codes for Single DMM and GCAS codes for batch DMM, are defined over a Galois field $\mathsf{GF}(q)$ of size $q\geq N$. These are inefficient for small Galois fields such as $\mathsf{GF}(2)$ and the integer residue ring $\mathbb{Z}_{p^{e}}$ due to the lack of invertible elements for interpolation. DMM over $\mathbb{Z}_{p^{e}}$ (such as $\mathbb{Z}_{2^{64}}$ ) is well-motivated in practice due to their direct compatibility with hardware. In this work, we construct efficient CDMM over the Galois ring $\mathsf{GR}(p^e,d)$ which is an extension ring over $\mathbb{Z}_{p^{e}}$ of degree $d$, particularly, $\mathsf{GR}(p,d)=\mathsf{GF}(p^d)$ is the Galois field and $\mathsf{GR}(p^e,1)=\mathbb{Z}_{p^e}$. We first give a general CDMM framework for the batch of $n$ matrix multiplications via the famous RMFE (Cascudo et al. Crypto'18). Compared with GCSA, our construction has a smaller recovery threshold by a factor of $1/n$. Next, we optimize EP codes via batch preprocessing of the input matrices. We give two types of Single CDMM, which can achieve almost the same performance as EP codes over a Galois field with size $q\geq N$. Finally, we present the experimental analysis of our CDMM on Galois rings.
    \end{abstract}

    \begin{IEEEkeywords}
        Matrix multiplication, Distributed computing, Coding, Reverse Multiplication Friendly Embedding.
    \end{IEEEkeywords}

    \section{Introduction}
        \IEEEPARstart{M}{atrix} multiplication is the core operation in linear algebra and has been one of the most practically applicable techniques in computer science, such as machine learning, scientific computing, and graph processing. Many such applications require processing terabytes or even petabytes of data. As an efficient solution for processing massive data sets, distributed computing divides large computing tasks into smaller subtasks and outsources them to a set of distributed server nodes. However, distributed computing faces critical issues related to communication load and straggler effects, i.e, the effect caused by some computing nodes which run unintentionally slower than others, thereby increasing the overall time needed to complete the computing tasks.

        Coded distributed matrix multiplication (CDMM for short) efficiently mitigates the straggler effects in distributed computing. In the framework of CDMM, a master node (or a user) distributes matrix multiplication across $N$ worker nodes (or servers) through a coding scheme, such that the response from any $R$ worker nodes ($R$ is called the recovery threshold) is sufficient for the master node to recover the desired product. Among all distributed matrix multiplication (DMM) schemes, CDMM has the advantage of theoretical analysis and flexibility in the recovery threshold. The main metrics of interest for CDMM include the recovery threshold $R$, the complexities of encoding, decoding and each worker node's computation, and the communication cost to upload subtasks to $N$ worker nodes and download data from $R$ successful worker nodes. 
        
        Many codes have been proposed in CDMM to optimize different metrics. According to the way of matrix-partition before encoding, there are Polynomial codes \cite{DBLP:conf/nips/YuMA17}, MatDot and PolyDot codes \cite{DBLP:journals/tit/DuttaFHJCG20}, and Entangled polynomial codes (EP codes for short) \cite{Entangle}, where EP codes are a unified framework of Polynomial and MatDot codes. Besides, Cross Subspace Alignment Codes (CSA codes for short) \cite{CSAcodes} are used for distributed computing of batch matrix multiplications, aiming to improve communication cost. In the same paper, by combining CSA and EP codes, the generalized CSA (GCSA) codes have been proposed to add more flexibility to the recovery threshold and each worker node's computation.

        To the best of our knowledge, most known CDMM schemes were restricted on a Galois field $\mathsf{GF}(q)$, where $q$ is the size of the field. Moreover, in state-of-the-art CDMM constructions, such as EP codes via matrix-partitioning or GCSA codes via batch-progressing and matrix-partitioning, the distributed order $N$ (i.e., the total number of worker nodes) is at most the field size, i.e., $q\geq N$. The integer residue ring $\mathbb{Z}_{2^{32}}$ or $\mathbb{Z}_{2^{64}}$ is well-motivated in practice as it is directly compatible with computation in real-life programming and computer architectures such as CPU words. However, unlike the Galois field $\mathbb{F}_{2^e}$ of size $2^e$, the lack of sufficient invertible elements under subtraction in $\mathbb{Z}_{2^{e}}$ brings non-trivial difficulties in designing CDMM schemes over $\mathbb{Z}_{2^{e}}$ since the polynomial interpolation or solving linear equations in the decoding process cannot work well in this case (the same is also true for small Galois fields $\mathsf{GF}(q)$ with $q<N$ for a given distributed order $N$). Let $\mathsf{GR}(2^e,m)$ be the extension Galois ring over $\mathbb{Z}_{2^{e}}$ with extension degree $m$ (refer to Section~\ref{sec:2.1}). By taking $m>\log N$, then $\mathsf{GR}(2^e,m)$ contains at least $N$ invertible elements supporting Lagrange interpolation. Thus, a trivial way is to embed matrices with data in $\mathbb{Z}_{2^e}$ or small Galois fields $\mathsf{GF}(q)$ into matrices with data in a large Galois ring to perform CDMM for large distributed order $N$, which we call a plain CDMM over the integer residue ring $\mathbb{Z}_{2^e}$. However, this will result in significant computation and communication overhead as the operations in the extension ring are much more expensive than $\mathbb{Z}_{2^e}$ (see Lemma~\ref{Lem:3.1} in Section~{\ref{sec.3.1} for details). 

        \subsection{Our Contribution}
            In this work, we study CDMM over a general Galois ring $\mathsf{GR}:=\mathsf{GR}(p^e,d)$ which is the extension ring of integer residue ring $\mathbb{Z}_{p^e}$ of degree $d$. Particularly, $\mathsf{GR}(2^e,1)=\mathbb{Z}_{2^{e}}$ is the integer residue ring and $\mathsf{GR}(p,d)=\mathsf{GF}(p^d)$ is the Galois field of size $p^d$. We give constructions of CDMM for large-scale matrices with data in $\mathsf{GR}$, which are as efficient as one can achieve over Galois fields for the same size of input matrices and distributed order. The novelty of our construction consists of batch preprocessing and using the famous Reverse Multiplication Friendly Embedding (RMFE for short, refer to Section~\ref{sec:2.3} for details). Assume $\{A_1,A_2,\dots,A_n\}$ and $\{B_1,B_2,\dots,B_n\}$ are two tuples of matrices after batch preprocessing of two given matrices $A$ and $B$ defined over $\mathsf{GR}$. By using the maps in RMFE, the master node first packed $\{A_i\}_{i=1}^n$ and $\{B_i\}_{i=1}^n$ into two matrices $\mathcal{A}$ and $\mathcal{B}$ defined over an extension Galois ring that has sufficient many invertible elements supporting interpolation. After perform CDMM over the extension ring with input $\mathcal{A}$ and $\mathcal{B}$, the master then unpacked the product $\mathcal{A}\mathcal{B}$ to get the batched products $\{A_iB_i\}_{i=1}^n$. Our results greatly change the landscape of CDMM over Galois rings, and $\mathbb{Z}_{2^{e}}$ in particular. Specifically,

            \begin{itemize}
                \item We first construct efficient distributed batch matrix multiplication (DBMM for short) based on EP codes over a general Galois ring $\mathsf{GR}$, i.e., given two tuples of matrices $(A_1,A_2,\dots,A_n)$ and $(B_1,B_2,\dots,B_n)$ with $A_i\in\mathsf{GR}^{s\times r}$ and $B_i\in\mathsf{GR}^{r\times t}$ for $i=1,2,\dots,n$, the goal is to compute the batch products $(A_1B_1, A_2B_2, \dots, A_nB_n)$. We evaluate the performance of our DBMM in terms of recovery threshold, encoding and decoding complexity, computation of each worker node, and the communication cost (all are counted in the number of operations or elements in $\mathsf{GR}$). For input matrices of the same size and maintaining the same distributed order $N$, our coded DBMM (CDBMM for short) over $\mathsf{GR}$ can achieve the same performance as EP codes over the Galois field $\mathsf{GF}({p^{(\log_pN)/{d}}})$ which is the smallest Galois field to implement EP codes with distributed order $N$. In other words, this is the best one can perform DBMM based on EP codes over a Galois ring with a given distributed order $N$. Please refer to our Theorem~\ref{thm 3.2} in Section~\ref{sec:3.2} for details. Besides, as EP code is the unified framework of Polynomial and Matdot codes, our CDBMM can be applied to these two codes.
                \item Next, we optimize Single matrix multiplication based on EP codes via different batch preprocessing of the input matrices $A, B$. We give two types of Single CDMM, called $\EP$-I and $\EP$-II. The first one has optimal metrics in terms of encoding complexity, upload costs, and each worker node's computation, while the second one has optimal metrics in terms of decoding complexity, download cost, and each worker node's computation. Please refer to our Corollaries~\ref{cor:mat} and \ref{cor:poly} in Section~\ref{sec:4} for details.
                \item Particularly, our CDBMM framework is suitable for any small Galois field, i.e., $\mathsf{GR}(p,d)=\mathsf{GF}({p^d})$. Given a distributed order $N$, for a small Galois field $\mathsf{GF}({p^d})$ with $p^d<N$, we have a CDMM scheme over $\mathsf{GF}(p^d)$ that is almost as efficient as the CDMM over the large Galois field $\mathsf{GF}({p^{\log_p N}})$ which is the smallest Galois field to perform EP codes with distributed order $N$. 
                \item At last, we implement our new constructions of CDMM over Galois ring $\mathsf{GR}(p^e,d)$ based on Entangled polynomial codes, Matdot codes, and Polynomial codes. An explicit data comparison of these codes with plane CDMM is presented in Section~\ref{sec:3.3}.
            \end{itemize}
            
            In this work, we only consider efficient computation of single and batch matrix multiplication. Many excellent works also consider secure and private distributed matrix multiplication \cite{Entangle, DBLP:journals/tifs/AliasgariSK20,secureDMM,ASK20,GASP}. Our CDMM based on Entangled polynomial codes over Galois ring $\mathsf{GR}(p^e,d)$ can be extended to secure and private computation and we left it as a future work.

        \subsection{Related work and comparisons}

            Coded distributed batch matrix multiplication and also secure CDBMM have been considered in many works \cite{CSAcodes,CDBMMcapacity,SCDBMM,Entangle}. The state-of-the-art DBMM is the generalized CSA codes (GCSA for short) \cite{CSAcodes}. To differentiate our CDBMM from others, we name it by Batch-$\EP$. According to our results presented in Theorem~\ref{thm 3.2}, we make a comparison Table 1 in Section~\ref{sec:3.2} which shows that, for the batch multiplication of $n$ pairs of matrices, Batch-$\EP$ has a smaller recovery threshold than GCSA by a factor of almost $\frac{1}{2n}$ when keeping the same communications and computational costs (i.e., $\kappa=n$ in Table 1). For the case that GCSA codes get the best recovery threshold (i.e., $\kappa=1$ in Table 1) which is still $n$ times the Recovery threshold in Batch-$\EP$, then the communication and computation costs are about $(\log_p N)/d$ times the corresponding costs of Batch-$\EP$. For $\mathsf{GR}=\mathsf{GR}(2^e,1)=\mathbb{Z}_{2^e}$, then the multiple $(\log_p N)/d=\log N$.
            
            As we have mentioned before, the single CDMM over a small Galois field for a large distributed order $N$ greater than the field size $q$ cannot work. Previously, some related work \cite{AGSCDM,DBLP:journals/corr/abs-2408-01806,SDMMAG} using algebraic geometry codes with at least $N$ rational points to design analog CDMM over $\mathsf{GF}({q})$, while our framework still use codes defined from univariate polynomial in $\mathsf{GF}(q)[x]$. By some fast algorithms for polynomial evaluation and interpolation over Galois ring \cite{von2003modern}, the encoding and decoding complexity in our CDMM are much faster than those based on AG codes. Moreover, due the the genus penalty in AG codes, CDMM based on AG codes also has a larger recovery threshold than univariate polynomial codes. Please refer to Remark~\ref{rmk:AGcomp} for detailed comparisons. 
      
            This paper is organized as follows. Section 2 introduces some preliminary about Galois rings and RMFE. The general framework of our coded distributed batch matrix multiplication via RMFE and comparisons with CDBMM based on GCSA codes are presented in section 3. Section 4 presents the main results of the Single CDMM using different batch prepossessing and comparisons with CDMM based on algebraic geometry codes. Section 5 presents experimental data of our Single CDMM schemes.

    \section{Preliminaries}
        \subsection{Notations.} Let \(p\) be a prime, and \(d\), \(e\), \(m\), \(n\), \(t\), \(r\), and \(s\) be positive integers. Generally, \(p^e\) represents the characteristic of the rings we consider, \(d\) and \(m\) represent the degrees of certain ring extensions, \(n\) denotes the dimension of the vectors to be packed, and \(t\), \(r\), and \(s\) represent the dimensions of the matrices involved in matrix multiplication. Moreover, in the CDMM schemes, we always denote $N$ the total number of worker nodes, $R$ the recovery threshold, and $u,v,w$ the divisors of $t,s,r$, respectively. By our settings, $R\leq N$ and $m\leq \log N$. 
        
        To denote algorithmic complexity, we will use the common big-O notation $O(\cdot)$ and big-theta notation $\Theta(\cdot)$. Moreover, for explicit comparisons, we use the soft O notation $\tilde{O}(\cdot)$ to denote the omission of $\log\log N$ terms, for example, $O(N\log^2 N\log\log N)=\tilde{O}(N\log^2 N)$. Unless otherwise specified, we always assume that all complexity is calculated in the number of elements or operations in the underlying Galois ring $\mathsf{GR}$.
        
        Vectors are denoted in boldface, and element-wise multiplication of vectors is represented by \(\mathbf{a} \star \mathbf{b}\). Matrices are denoted by capital letters, and \(A[i,j]\) represents the element in the \(i\)-th row and \(j\)-th column of matrix \(A\). Specifically, \(A[i,\cdot]\) and \(A[\cdot,j]\) denote the \(i\)-th row and \(j\)-th column of matrix \(A\) respectively.

        \subsection{Galois Rings.} \label{sec:2.1}
            Let $p$ be a prime and $\mathbb{Z}_{p^e}$ be the integer residue classes modulo $p^e$. The Galois ring is an extension of $\mathbb{Z}_{p^e}$ and has characteristic $p^e$. Specifically, let $d\geq 1$ be an integer and $f(x)\in\mathbb{Z}_{p^e}[x]$ be a monic polynomial of degree $d$ such that $\bar{f}(x):=f(x) \bmod p$ is an irreducible polynomial over the finite field $\mathsf{GF}(p)$. The Galois ring of extension degree $d$ over $\mathbb{Z}_{p^e}$ is defined as
                \[\mathsf{GR}(p^e,d):=\mathbb{Z}_{p^e}[x]/(f(x)).\]
            Let $\xi=x+(f(x))\in\mathsf{GR}(p^e,d)$ be a root of $f(x)$. Then any $\alpha\in\mathsf{GR}(p^e,d)$ can be expressed uniquely in the form of $\alpha=a_0+a_1\xi+\dots+a_{d-1}\xi^{d-1}$, $a_i\in\mathbb{Z}_{p^e}$. We refer to \cite{wan2003lectures} for more details about Galois rings. In the following, we briefly use $\mathsf{GR}$ to denote $\mathsf{GR}(p^e,d)$.

            Let $(p)$ be the unique maximal ideal of $\mathsf{GR}(p^e,d)$. Then there is a field isomorphism 
                \[\tau:\ \mathsf{GR}(p^e,d)/(p)\cong \mathsf{GF}({p^d})=\mathsf{GF}(p)[x]/(\bar{f}(x)).\]
            Thus for any $\alpha\in\mathsf{GR}(p^e,d)$, $\alpha$ is a unit if and only if $\alpha \mod p\neq 0$. By the isomorphism $\tau$ and multiplicative structure of $\mathsf{GF}({p^d})^{*}$, there exists an element $\zeta\in\mathsf{GR}(p^e,d)$ of multiplicative order $p^{d}-1$. Let $T=\{0,1,\zeta,\zeta^2,\ldots,\zeta^{p^d-2}\}$, then $T$ is a set of representatives for equivalence classes modulo $(p)$. Hence, for any two distinct elements $\alpha,\alpha'\in T$, we have $\alpha-\alpha'$ is invertible. We call a subset of $\mathsf{GR}$ with this property, i.e., the subtraction of any two distinct elements in it is invertible, an \textbf{exceptional} set. 

            Due to the exceptional property of $T$, the following Lagrange interpolation formula works well on $T$: Assume $\{x_1,x_2,\dots,x_n\}\subset T$. Let $\lambda_i=\prod_{j\neq i}(x_i-x_j)^{-1}$ for $i\in[n]$ and
                  \[f(x) := \sum_{i=1}^n y_i \lambda_i \prod_{\substack{j=1 \\ j \neq i}}^n (x - x_j)\]
                then $f(x_i)=y_i$ for $i=1,2,\dots,n$.
            
           Moreover, there exist fast algorithms for multipoint evaluation and interpolation on $T$.
           \begin{Lem}[\cite{von2003modern}]\label{prop:exp} 
                Let $\mathsf{GR}=\mathsf{GR}(p^e,d)$. Assume $n\leq p^d$ is a positive integer and $\{x_1,x_2,\dots,x_n\}$ is a tuple of $n$ elements in $T$.
                \begin{itemize}
                    \item[(i)] Given any nonzero polynomial $f(x)\in\mathsf{GR}[x]$ of degree less than $n$, then $\{f(x_1),f(x_2),\dots,f(x_n)\}$ can be computed with $O(n\log^2 n\log\log n)$ operations in $\mathsf{GR}$. 
                    \item[(ii)] Conversely, given any \(y_1,\dots, y_n \in \mathsf{GR}\), there is a unique polynomial \(f(x) \in \mathsf{GR}[X]\) of degree less than \(n\) such that \(f(x_i)=y_i\); furthermore, such an $f(x)$ can be computed with $O(n\log^2 n\log\log n)$ operations in $\mathsf{GR}$.
                \end{itemize}
            \end{Lem}
            
            \begin{proof}
                Please refer to \cite[Corollary~10.8 and ~10.12]{von2003modern} for the proof of (i) and (ii), respectively.
            \end{proof}

        \subsection{Reverse Multiplication Friendly Embedding}\label{sec:2.3}
            RMFE was initially proposed in \cite{DBLP:conf/crypto/CascudoCXY18} to reduce the amortized cost, but it only works for finite fields and only supports single multiplication. The work of \cite{DBLP:conf/crypto/CramerRX21} extends RMFE to a more general Galois ring, and multiple multiplications are supported in \cite{DBLP:conf/asiacrypt/EscuderoHLXY23}. In this work, we focus on a single multiplication case over a general Galois ring. Let us first present the definition of RMFE and some useful theorems.
            \begin{Def}[RMFE, {\cite[Definition 17]{DBLP:conf/crypto/CramerRX21}}]
                Let \(\mathsf{GR}(p^e, d)\) be a Galois ring of degree \(d\) and characteristic \(p^e\). A pair \((\phi, \psi)\) is called an $(n, m)$-RMFE over \(\mathsf{GR}(p^e, d)\) if \(\phi : \mathsf{GR}(p^e, d)^n \to \mathsf{GR}(p^e, dm)\) and \(\psi : \mathsf{GR}(p^e, dm) \to \mathsf{GR}(p^e, d)^n\) are two \(\mathsf{GR}(p^e, d)\)-linear maps satisfying
                \[
                    \mathbf{x} \star \mathbf{y} = \psi(\phi(\mathbf{x}) \cdot \phi(\mathbf{y}))
                \]
                for all \(\mathbf{x}, \mathbf{y} \in \mathsf{GR}(p^e, d)^n\).
            \end{Def}

            Many works consider the existence of an $(n,m)$-RMFE over a Galois field. Since we consider fast algorithms for CDMM, we adopted the one in \cite{cramer2020complexity}, which has a quasi-linear time algorithm for RMFE. Then, by the work of \cite{DBLP:conf/crypto/CramerRX21}, we have similar results over a general Galois ring.

            \begin{Lem}[\cite{cramer2020complexity,DBLP:conf/crypto/CramerRX21}]\label{lem:rmfe1}
               Let $\mathsf{GR}(p^e,d)$ be a Galois ring. There exists an algorithm to generate a family of $(n_i,m_i)$-RMFEs over $\mathsf{GR}(p^e,d)$ that can be computed in time $O(m_i\log^2 m_i \log\log m_i)$ as $n_i\rightarrow \infty$ and $\lim_{i\rightarrow \infty}\frac{m_i}{n_i}\leq C$ for some constant $C$.
            \end{Lem}

            \begin{Rmk}
                For $\mathsf{GR}(p^e,1)=\mathbb{Z}_{2^e}$, the upper bound $C$ of the rate of $m/n$ was shown to be $C\approx 5$ \cite{DBLP:conf/crypto/CascudoCXY18}; for a general Galois ring, then $C\approx 8+\frac{16}{\sqrt{p^d}-1}$ by \cite{cramer2020complexity}. In the following, we do not care about the explicit value of $C$ and assume that for $m=O(n)$ there exists an algorithm for $(n,m)$-RMFE with time complexity $O(m\log^2 m\log\log m)$.
            \end{Rmk}

            In Section \ref{sec:3.3}, we need the following concatenation results of RMFEs.
    
            \begin{Lem}[Concatenation,   {\cite[Lemma 5]{DBLP:conf/crypto/CascudoCXY18}} and  {\cite[Lemma 4]{DBLP:conf/asiacrypt/EscuderoHLXY23}}]
                Assume that \((\phi_1, \psi_1)\) is an \((n_1, m_1)\)-RMFE over \(\mathsf{GR}(p^e, dm_2)\) and \((\phi_2, \psi_2)\) is an \((n_2, m_2)\)-RMFE over \(\mathsf{GR}(p^e, d)\). Then $(\phi=\phi_1\circ\phi_2,\psi=\psi_2\circ\psi_1)$ is an $(n_1n_2,m_1m_2)$-RMFE over $\mathsf{GR}(p^e,d)$, where \(\phi: \mathsf{GR}(p^e, d)^{n_1n_2} \to \mathsf{GR}(p^e, dm_1m_2)\) be given by
                \[\begin{split}
                   (\mathbf{x_1}, \dots, \mathbf{x}_{n_1}) & \mapsto (\phi_2(\mathbf{x_1}), \dots, \phi_2(\mathbf{x}_{n_1})) \in \mathsf{GR}(p^e, dm_2)^{n_1} \\
                    &\mapsto \phi_1(\phi_2(\mathbf{x_1}), \dots, \phi_2(\mathbf{x}_{n_1}))
                \end{split}
                \]
                and \(\psi : \mathsf{GR}(p^e, dm_1m_2) \to \mathsf{GR}(p^e, d)^{n_1n_2}\) is given by
                \[\begin{split}
                    \alpha &\mapsto \psi_1(\alpha) = (\mathbf{u}_1, \dots, \mathbf{u}_{n_1}) \in \mathsf{GR}(p^e, dm_2)^{n_1} \\
                    &\mapsto (\psi_2(\mathbf{u}_1), \dots, \psi_2(\mathbf{u}_{n_1})).\end{split}
                \]
            \end{Lem}

    \section{Coded distributed batch matrix multiplication via RMFE} \label{sec:bmm}
        In this section, we first introduce a general framework of coded distributed batch matrix multiplication using RMFE. Next, we analyze the amortized cost of the DBMM scheme within our framework and compare it with batch DMM based on GCSA codes.
        
        \subsection{The general framework} \label{sec.3.1}
            Let $\mathsf{GR}:=\mathsf{GR}(p^e,d)$ be an arbitrary Galois ring with characteristic $p^e$. As we have shown in Section~\ref{sec:2.1}, the exceptional set $T$ of $\mathsf{GR}$ has $p^d$ elements in total. When $d$ is small such as $d=1$, by applying known CDMM to $\mathsf{GR}$ for large-scale input matrices \( A \) and \( B \) with entries in \( \mathsf{GR} \), then the distributed order $N\leq p6d$. Thus it is inefficient to perform known CDMM over $\mathsf{GR}$ with small extension degree $d$. To solve this obstacle, we must embed $A$ and $B$ into an extension Galois ring $\mathsf{GR}'=\mathsf{GR}(p^e,dm)$ of extension degree $m$. The smallest $m$ is at least $m\geq\lceil\log_p(N)/d\rceil$, which is close to $O(\log N)$ for constant $p$ and $d$. As a result, a CDMM scheme performed in $\mathsf{GR}'$ will bring extra ${O}(m\log^2 m\log\log m)=\tilde{O}(\log N)$ computational and $O(m)=O(\log N)$ communication overhead since every $\alpha'\in\mathsf{GR}'$ can be viewed as a polynomial of degree less than $m$ in $\mathsf{GR}$. Fortunately, the additional cost can be amortized by RMFE when a batch of large-scale matrix multiplications need to be computed. The general framework is illustrated in Figure~\ref{fig:rmfebmm}.

            \begin{figure*}[h]
                \centering
                \includegraphics[scale=0.9]{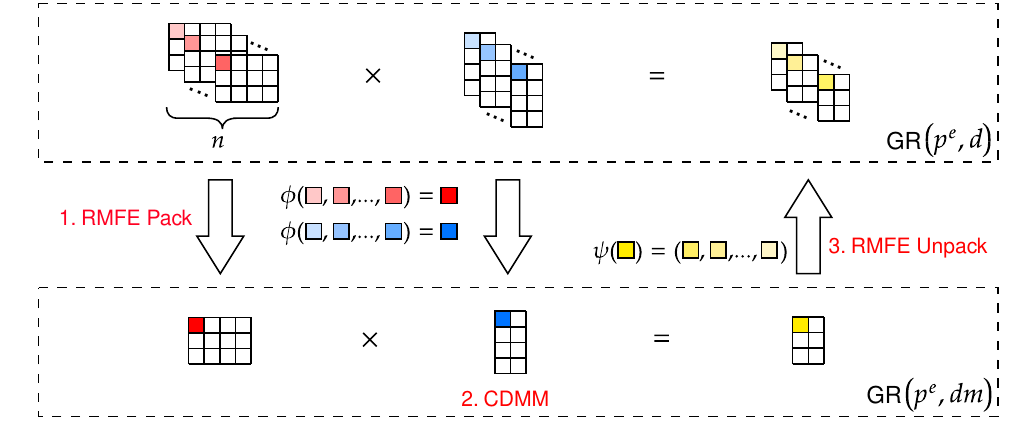}
                \caption{General Framework of Applying RMFE to Batch DMM}
                \label{fig:rmfebmm}
            \end{figure*}

            In the following, we will briefly use $\mathsf{GR}$ to denote the Galois ring $\mathsf{GR}(p^e,d)$ and $\mathsf{GR}_m$ to denote its extension ring $\mathsf{GR}(p^e, dm)$ of extension degree $m$. Assume \((A_1, \dots, A_n)\) and \((B_1, \dots, B_n)\) are two batches of large matrices, where $A_i \in \mathsf{GR}^{t \times r}$ and $B_i \in \mathsf{GR}^{r \times s}$. To compute \(C_i = A_i B_i\) for \(1 \leq i \leq n\) in parallel, we use an \((n, m)\)-RMFE over \(\mathsf{GR}\):
                \[\mathsf{GR}^n\xrightarrow{\phi} \mathsf{GR}_m\xrightarrow{\psi} \mathsf{GR}^n.\] 
            The existence of an $(n,m)$-RMFE is ensured by setting $m=O(n)$ as in Lemma~\ref{lem:rmfe1}. Specifically, we first group all elements at the same position from different matrices in the same batch into one vector, i.e.,
            \[\begin{split}
               &\mathbf{a}_{i,j} = (A_1[i,j],\dots,A_n[i,j]),\\
               &\mathbf{b}_{k,\ell} =(B_1[k,\ell],\dots,B_n[k,\ell]),   
            \end{split}
            \]
            where \(\mathbf{a}_{i,j}, \mathbf{b}_{k,\ell}\in \mathsf{GR}^n.\) Next, under the map \(\phi\), all vectors are mapped into the extension Galois ring \(\mathsf{GR}_m\):
            \[
                \alpha_{i,j}=\phi(\mathbf{a}_{i,j}), \quad\beta_{k,\ell}=\phi(\mathbf{b}_{k,\ell}).
            \]
            Therefore, two batches of matrices $\{A_i\}_{i=1}^n$ and $\{B_i\}_{i=1}^n$ are packed to two matrices $\mathcal{A}$ and $\mathcal{B}$ as follows:
            \[
                \mathcal{A} =
                \begin{pmatrix}
                    \alpha_{1,1} & \cdots & \alpha_{1,r} \\
                    \vdots & \ddots & \vdots \\
                    \alpha_{t,1} & \cdots & \alpha_{t,r}
                \end{pmatrix}
                , \quad
                \mathcal{B} =
                \begin{pmatrix}
                    \beta_{1,1} & \cdots & \beta_{1,s} \\
                    \vdots & \ddots & \vdots \\
                    \beta_{r,1} & \cdots & \beta_{r,s}
                \end{pmatrix}.
            \]
            Assume we have CDMM schemes over the extension ring $\mathsf{GR}_m$, such as EP codes. Then by applying CDMM with input $\mathcal{A}$ and $\mathcal{B}$, the master node get \(\mathcal{C} = \mathcal{AB}\). Finally, by using the map \(\psi\) elementwise to $\mathcal{C} $, i.e.,

            \begin{align}
                \psi(\mathcal{C}[i,\ell])&=(\psi(\sum_{j=1}^r \alpha_{i,j}\beta_{j,\ell}))=(\sum_{j=1}^r\psi( \alpha_{i,j}\beta_{j,\ell}))\nonumber\\
                    &=(\sum_{j=1}^r\psi( \phi(\mathbf{a}_{i,j})\phi(\mathbf{b}_{j,\ell}))=(\sum_{j=1}^r \mathbf{a}_{i,j}\star\mathbf{b}_{j,\ell})\nonumber\\
                    &=(\sum_{j=1}^r A_1[i,j]B_1[j,\ell], \dots, \sum_{j=1}^rA_n[i,j]B_n[j,\ell])\nonumber\\
                    &=(A_1[i,\cdot]B_1[\cdot,\ell], \dots, A_n[i,\cdot]B_n[\cdot,\ell])\nonumber\\
                    &=(C_1[i,\ell], \dots, C_n[i,\ell])\nonumber.
            \end{align}
            This yields the expected result $(C_1,C_2,\dots,C_n)$. In the following subsection, we will provide a detailed explanation of how the above framework can be applied to batch matrix multiplication based on EP codes.

        \subsection{Batch DMM based on EP codes} \label{sec:3.2}
            Let $\mathcal{A}$ and $\mathcal{B}$ be the packed matrices of $\{A_i\}_{i=1}^n$  and $\{B_i\}_{i=1}^n$ under $\phi$, respectively. By the CDMM based on EP codes \cite{DBLP:journals/tit/YuMA20}, \(\mathcal{A}\) and \(\mathcal{B}\) are partitioned into \(uw\) and \(vw\) submatrices as follows:
            \[
                \mathcal{A} =
                \begin{pmatrix}
                    \mathcal{A}_{11} & \cdots & \mathcal{A}_{1w} \\
                    \vdots & \ddots & \vdots \\
                    \mathcal{A}_{u1} & \cdots & \mathcal{A}_{uw}
                \end{pmatrix}
                , \quad
                \mathcal{B} =
                \begin{pmatrix}
                    \mathcal{B}_{11} & \cdots & \mathcal{B}_{1v} \\
                    \vdots & \ddots & \vdots \\
                    \mathcal{B}_{w1} & \cdots & \mathcal{B}_{wv}
                \end{pmatrix}
                ,
            \]
            where \(\mathcal{A}_{ij}\in \mathsf{GR}_m^{\frac{t}{u}\times \frac{r}{w}}\) and \(\mathcal{B}_{k\ell}\in \mathsf{GR}_m^{\frac{r}{w}\times \frac{s}{v}}\). Then, the matrix product \(C\) is computed as:
            \[
                \mathcal{C} =
                \begin{pmatrix}
                    \mathcal{C}_{11} & \cdots & \mathcal{C}_{1v} \\
                    \vdots & \ddots & \vdots \\
                    \mathcal{C}_{u1} & \cdots & \mathcal{C}_{uv}
                \end{pmatrix}
                ,
            \]
            where \(\mathcal{C}_{i\ell} = \sum_{j=1}^w \mathcal{A}_{ij} \mathcal{B}_{j\ell}\). 
            
            The master node first chooses \(N\leq p^{dm}\) pairwise distinct points \(\alpha_1,\dots,\alpha_N\) from the exceptional set of \(\mathsf{GR}_m\) and then constructs two polynomials (with matrix coefficients):
            \[\begin{split}
               & f(x) := \sum_{i=1}^{u} \sum_{j=1}^{w}  \mathcal{A}_{ij}x^{(i-1)w +j-1}, \\
               & g(x) := \sum_{k=1}^{w} \sum_{\ell=1}^{v}\mathcal{B}_{k\ell} x^{w-k + (\ell-1) uw}, \end{split}
            \]
            Then
            \[ \begin{split}
                h(x) &:= f(x)g(x) \\
                &=  \sum_{i=1}^{u} \sum_{j=1}^{w} \sum_{k=1}^{w} \sum_{\ell=1}^{v}\mathcal{A}_{ij} \mathcal{B}_{k\ell} x^{(i-1)w + (w-1 -k+j) + (\ell-1)uw}. \end{split}
            \]
            Next, the master node computes and sends the evaluations \(f(\alpha_i),g(\alpha_i)\) to the \(i\)-th worker node for \(1\le i\le N\). Then \(i\)-th worker node computes the small matrix multiplication of the two evaluations \(h(\alpha_i)=f(\alpha_i)g(\alpha_i)\) locally. The master node only needs to download the values of $h(\alpha_i)$ from any $\deg(h)+1$ successful worker nodes to interpolates \(h(x)\in \mathsf{GR}_m^{\frac{t}{u}\times\frac{s}{v}}[x]\). Then the product \(\mathcal{C}\) can be recovered from the coefficients of $h(x)$. Hence the recovery threshold is \(R=\deg(h) + 1 =uvw + w -1\).

            The encoding and decoding complexity, communication overhead (upload/download), and computational load at each worker are described in detail in \cite{DBLP:journals/tit/YuMA20}. For comparison, we present the results as follows. The computational complexity is counted as the total number of operations in $\mathsf{GR}=\mathsf{GR}(p^e,d)$, and the communication cost is counted as the total number of elements in $\mathsf{GR}$.

            \begin{Lem}[EP Codes \cite{DBLP:journals/tit/YuMA20}] \label{Lem:3.1}
                Given two matrices $\mathcal{A} \in\mathsf{GR}(p^e, d)^{t \times r}$ and $\mathcal{B} \in\mathsf{GR}(p^e, d)^{r \times s}$, let $u,v,w$ be the partition parameters that are factors of $t,s,r$, respectively. By applying EP codes of length $N$ to compute $\mathcal{C} = \mathcal{A}\mathcal{B}$, the recovery threshold $R=uwv + w - 1$. Let $m=\lceil (\log_pN)/d\rceil$. The communication costs are
                \begin{itemize}
                    \item[--] Upload: $O\left(\frac{tvr+sur}{uvw}) Nm\right)$. 
                    \item[--] Download: $O\left(\frac{ts}{uv} R m\right)$.
                \end{itemize}
                Moreover, the computational complexities are 
                \begin{itemize}
                    \item[--] Encoding: $\tilde{O}\left(\left(\frac{tvr+sur}{uvw}\right)mN\log^2 N \right)$.
                    \item[--] Decoding: $\tilde{O}\left(\frac{ts}{uv} mR\log^2 R \right)$.
                \item[--] The computation of each worker node: $\tilde{O}\left(\frac{trs}{uwv}m\right)$.
                \end{itemize}
            \end{Lem}
            \begin{proof}
                The recovery threshold is the same as the original EP codes performed in a Galois field. We set $m=\lceil \frac{\log_p N}{d}\rceil$ to support EP codes works well on the extension Gaois ring $\mathsf{GR}_m=\mathsf{GR}(p^e, dm)$. Let $\mathbb{A}=\{x_1,x_2,\dots,x_N\}$ be a subset of the maximal exception set $T$ of $\mathsf{GR}_m$. By the fast polynomial evaluation and interpolation over $\mathbb{A}$, the encoding and decoding cost $O\left(\left(\frac{t}{u} + \frac{s}{v}\right)\frac{r}{w}N\log^2 N\log\log N \right)$ and $O\left(\frac{ts}{uv} R \log^2 R\log\log R \right)$ operations in $\mathsf{GR}_m$, respectively. As $[\mathsf{GR}_m: \mathsf{GR}]=m$, each operation in $\mathsf{GR}_m$ costs at most $O(m\log^2 m\log\log m)$ operations in $\mathsf{GR}$ \cite{von2003modern}. The desired encoding and decoding complexity are derived by omitting $\log\log N$ terms in the soft-O notation. The extra multiple $O(m)$ in the communication cost is also caused by the overhead of representing an element in $\mathsf{GR}_m$ as combinations of $m$ elements in $\mathsf{GR}$. 
            \end{proof}

            Let Batch-$\EP$ denote the batch CDMM described as in Figure 1 by combining EP codes over the Galois ring. Then, according to the computational complexity of RMFE andlemma~\ref{Lem:3.1}, we get the computational and communication complexity of our Batch Distributed Matrix Multiplication as follows: 

            \begin{Thm}[Batch DMM via EP Codes]\label{thm 3.2}
                Assume \((A_1, \dots, A_n)\) and \((B_1, \dots, B_n)\) are two batches of matrices, where each $A_i \in\mathsf{GR}(p^e, d)^{t \times r}$ and $B_i \in\mathsf{GR}(p^e, d)^{r \times s}$. Assume $m=\lceil (\log_pN)/d\rceil=O(n)$. Then the recovery threshold of $\mathsf{Batch}$-$\EP$ is $R=uwv + w - 1$. Moreover, the amortized communication cost (per matrix multiplication) is summarized as follows
                \begin{itemize}
                    \item[--] The upload cost is $O\left(\left(\frac{tr}{uw} + \frac{rs}{wv}\right) N\right)$. 
                    \item[--] The download cost is $O\left(\frac{ts}{uv} R\right)$.
                \end{itemize}
                The amortized computational complexities are
                \begin{itemize}
                    \item[--] Encoding: $\tilde{O}\left(\left(\frac{trv+sru}{uvw}\right)N \log^2 N \right)$.
                    \item[--] Decoding: $\tilde{O}\left(\frac{ts}{uv} R\log^2 R \right)$.
                \item[--] The computation of each worker node: $\tilde{O}\left(\frac{trs}{uwv}\right)$.
                \end{itemize}
            \end{Thm}
            \begin{proof}
                According to the Batch-$\EP$ framework illustrated in Figure 1, the recovery threshold, communication cost, and computational complexity for each worker node are equivalent to those in EP codes over $\mathsf{GR}_m$. By setting \( m = \lceil (\log_p N)/d \rceil = O(n) \), the amortized communication cost (and each worker's computational cost) per matrix multiplication is reduced by  \( n \), yielding the desired efficiency gains.

                All that remains is to first calculate the total encoding and decoding complexity by adding the computational cost of RMFE and then amortize it by $n$. 
                For encoding, the master node needs to compute \( \phi(\mathbf{a}_{i,j}) \) and \( \phi(\mathbf{b}_{k,\ell}) \) for \( [i,j,k,\ell] \in [1,t] \times [1,r] \times [1,r] \times [1,s] \). Since the complexity to compute one function \( \phi \) is \( O(m \log^2 m\log\log m) \). Hence, the amortized complexity to compute \( \mathcal{A} \) and \( \mathcal{B} \) are \( O((t+s)rm \log^2 m\log\log m)=\tilde{O}((t+s)r\log N)\) since $m=O(\log N)$. Then, by the encoding complexity of EP codes as in Lemma~\ref{Lem:3.1}, the amortized encoding complexity per matrix multiplication in Batch-$\EP$ is 
                \[\begin{split}
                &\tilde{O}\left(\left(\frac{trv+sru}{uvw}\right) N \log^2 N + (t+s)r \log N\right) \\
                &= \tilde{O}\left(\left(\frac{trv+sru}{uvw}\right)  N \log^2 N \right),    
                \end{split}
                \]
                
                During the decoding process in Batch-$\EP$, the master node needs to compute $$\psi(\mathcal{C})=\big(\psi(\mathcal{C}[i,\ell])\big)_{1\leq i\leq s, 1\leq \ell\leq t},$$ where
                \( \mathcal{C} = \mathcal{AB}\in \mathsf{GR}(p^e, dm)^{t \times s},\) the  complexity to compute one function \( \psi \) is \( O(m \log^2 m\log\log m) \). Hence, the amortized complexity to the matrix \( \psi(\mathcal{C}) \) is \(\tilde{O}(ts m)\). By the decoding complexity of EP codes, the amortized decoding complexity is \( \tilde{O}\left( \frac{ts}{uv} R \log^2 R + ts m\right) = \tilde{O}\left( \frac{ts}{uv} R \log^2 R \right) \).
            \end{proof}

            Let $N, R, n$ and $u,v,w$ be defined as in Theorem~\ref{thm 3.2}. We now compare our results with GCSA codes for batch matrix multiplication. The smallest Galois field $\mathsf{GF}(q)$ to support GSCA codes must have size $q\geq N+n$. Thus the smallest Galois ring $\mathsf{GR}_m$ must have exceptional set of size $p^{dm}\geq N+n$ to support GSCA, while in our Batch-$\EP$, we only need $p^{dm}\geq N$. To keep consistency, we assume $\mathsf{GR}_m$ is sufficiently large for GCSA and Batch-$\EP$ and take $n=\Theta(m)$. In Table 1, $\kappa$ is a positive divisor of $n$ defined as in \cite{CSAcodes}, and we omit terms of $\log\log N$ in the big $\tilde{O}$ notation. We use bold items to indicate better parameters   
            \begin{table*}[ht]
                \centering
                \caption{Comparison of Batch-coded matrix multiplication over Galois ring with GCSA code.}
                \begin{tabular}{|c|c|c|}
                    \hline
                       &  GCSA\cite{CSAcodes} & Batch-$\EP$ \\ \hline
                     Recovery threshold $R$  & $uvw(n+\kappa-1)+w-1$ & $\mathbf{uvw+w-1}$ \\ \hline
                     Upload complexity & $(\frac{trv+sru}{uvw} )\frac{n}{\kappa}N$ & $\mathbf{(\frac{trv+sru}{uvw})N}$ \\ \hline
                     Download complexity & $\frac{ts}{uv}R$ & $\frac{ts}{uv}R$\\ \hline
                    Worker node's computation & $\tilde{O}(\frac{trs}{uvw}\frac{n}{\kappa})$ & $\mathbf{\tilde{O}(\frac{trs}{uvw})}$ \\ \hline
                     Encoding complexity & $\tilde{O}((\frac{trv+sru}{uvw})\frac{n}{\kappa}N\log^2 N)$ & $\mathbf{\tilde{O}((\frac{trv+sru}{uvw} )N\log^2 N)}$\\ \hline
                     Decoding complexity & $\tilde{O}(\frac{ts}{uv}\frac{n}{\kappa}R\log^2 R)$ & $\mathbf{ \tilde{O}(\frac{ts}{uv}R\log^2 R)}$ \\ \hline
                \end{tabular}
                \label{tab:my_label}
            \end{table*}  

            \begin{Rmk}
               Additionally, we can use Polynomial codes \cite{DBLP:conf/nips/YuMA17} or MatDot codes \cite{DBLP:journals/tit/DuttaFHJCG20} in place of EP codes within our CDBMM framework.  The encoding and decoding complexity, communication overhead (upload/download), and computational load at each worker are similar to those in Theorem~\ref{thm 3.2}. When using Polynomial codes, $w = 1$, while with MatDot codes, $u = v = 1$.
            \end{Rmk}

    \section{Optimize Single Distributed Matrix Multiplication via RMFE} \label{sec:3.3}
        Despite the batch scenario, the technique can also be applied into optimizing the single matrix multiplication. Notice that in the first step of code-based DMM, the master node splits original matrices into submatrices. After the partition, to compute the product of original two matrices, a series of product of submatrices are needed to be computed in parallel. This observation suggests that for a single large matrix multiplication, we can split the original matrix to “manually” create a batching scenario, and then apply RMFE to optimize it. Then, we will delve into optimizing a single DMM via RMFE and analyze its associated overhead. 
       
        We still discuss the constructions over a general Galois ring $\mathsf{GR}=\mathsf{GR}(p^e, d)$ and let $\mathsf{GR}_m$ be its extension ring of extension degree $m$. In the following construction of $\EP$-I for single matrix multiplication of $A$ and $B$, we partition $A$ and $B$ as the way in MatDot codes, then $AB=\sum_{i=1}^n A_i B_i$ is the inner product. We first compute batch matrix multiplications $A_1B_1, A_2B_2,\dots, A_nB_n$ by using our batch-$\EP$ and then add them to get $AB$. 

        \medskip

        \fbox{
            \begin{minipage}{0.95\linewidth}
                \begin{center} {\bf $\EP$-I}\end{center}
                \begin{itemize}
                    \item[--] Input: \( A \in \mathsf{GR}^{t \times r} \) and \( B \in \mathsf{GR}^{r \times s} \).
                    \item[--] Output: $AB$.
                    \item[(1)] Let $N$ be the total number of Worker nodes and $m=\lceil \log_p(N)/d\rceil$. Assume $n=\Theta(m)$ is a divisor of $r$. Then \( A \) and \( B \) can be partitioned as follows
                                \[
                                    A = \begin{pmatrix}
                                        A_1 & \cdots & A_n
                                    \end{pmatrix}, \quad 
                                    B =
                                    \begin{pmatrix}
                                        B_1 \\
                                        \vdots \\
                                        B_n
                                    \end{pmatrix}
                                    ,
                                \]
                                where \(A_i\in\mathsf{GR}^{t\times \frac{r}{n}}\) and \(B_i\in\mathsf{GR}^{\frac{r}{n}\times s}\).
                    \item[(2)] Call Batch-$\EP$ with input $\{A_i\}_{i=1}^n$ and $\{B_i\}_{i=1}^n$. Then output $\sum_{i=1}^n C_i$, where $C_i=A_iB_i$ for $i=1,2,\dots,n$.
                \end{itemize}
            \end{minipage}
        }

        \medskip

        \begin{Cor}\label{cor:mat}
            Let $A, B, N, m, n$ be defined as in $\EP$-I. Assume $u,v,w$ are divisors of $s,t$, and $\frac{r}{n}$, respectively, such that $R:=uvw+w-1\leq N$. The recovery threshold in the above $\EP$-I for single matrix multiplication is $R$. Moreover, the communication complexities are 
            \begin{itemize} 
                \item[--] Upload: $O\left((\frac{trv+sru}{uvw} ) N\right)$.
                \item[--] Download: $O\left(\frac{ts}{uv}mR\right)$. 
            \end{itemize}
            The computational costs are
            \begin{itemize}
                \item[--] Encoding: $\tilde{O}(\left(\frac{trv+sru}{uvw} \right) N\log^2 N )$.
                \item[--] Decoding: $\tilde{O}(\frac{ts}{uv}m R\log^2 R)$.
                \item[--] The computation of each worker node: $\tilde{O}\left(\frac{tsr}{uvw}\right)$.
            \end{itemize}
        \end{Cor}
        \begin{proof}
            The recovery threshold, upload and download complexities, as well as the encoding and computational complexity for each worker node, follow directly from Theorem~\ref{thm 3.2} by substituting \( r \) with \( r/n \) and assuming \( m = O(n) \).
            Considering the decoding procedure, we need to account for the summation cost of \(\sum_{i=1}^n C_i\), which requires at most \(\tilde{O}\left(n \frac{st}{uv}\right)\) operations and is dominated by the decoding complexity. Thus, the total decoding complexity is \(\tilde{O}\left(\frac{ts}{uv} m R \log^2 R\right)\).
        \end{proof}

        The batch preprocessing of matrices \( A \) and \( B \) in $\EP$-I follows the partition way in MatDot codes. Another approach of batch-possessing can also follow the type of splitting in Polynomial codes. Thus, we give the second construction of Single $\EP$-II in the following. It is easy to see that Batch-$\EP$ based on EP type partitioning is composed of $\EP$-I and $\EP$-II, we will not go into details.

        Let \( N \) denote the number of worker nodes in EP codes, with \( m = \lceil \log_p(N)/d \rceil\). Our second construction for single matrix multiplication is based on batch processing via partitioning in Polynomial codes and involves two applications of RMFE. Suppose \((\phi_1, \psi_1)\) is an \((n, \sqrt{m})\)-RMFE defined over \(\mathsf{GR}\), and \((\phi_2, \psi_2)\) is an \((n, \sqrt{m})\)-RMFE defined over \(\mathsf{GR}_{\sqrt{m}}\), as follows:

        \[
        \mathsf{GR}^n \xrightarrow{\phi_1} \mathsf{GR}_{\sqrt{m}} \xrightarrow{\psi_1} \mathsf{GR}^n,
        \]
        \[
        \mathsf{GR}_{\sqrt{m}}^n \xrightarrow{\phi_2} \mathsf{GR}_m \xrightarrow{\psi_2} \mathsf{GR}_{\sqrt{m}}^n.
        \]

        This approach enables efficient batch processing for single matrix multiplication by leveraging structured partitioning and the properties of RMFE across distributed nodes. 

        \medskip

        \fbox{
            \begin{minipage}{0.95\linewidth}
                \begin{center} {\bf $\EP$-II}\end{center}
                \begin{itemize}
                    \item[--] Input: \( A \in \mathsf{GR}^{t \times r} \) and \( B \in \mathsf{GR}^{r \times s} \).
                    \item[--] Output: $AB$.
                    \item[(1)] Let $N$ be the number of worker nodes and $m=\left\lceil\log_p (N)/d\right\rceil$. Assume $n=\Theta(\sqrt{m})$ is a common divisor of $s$ and $t$. Then \( A \) and \( B \) can be partitioned as follows
                                \[A =
                                    \begin{pmatrix}
                                        A_1 \\
                                        \vdots \\
                                        A_n
                                    \end{pmatrix},\ 
                                    B = \begin{pmatrix}
                                        B_1 & \cdots & B_n
                                    \end{pmatrix}, 
                                \]
                                where \(A_i\in\mathsf{GR}^{\frac{t}{n}\times r }\) and \(B_j\in\mathsf{GR}^{r\times \frac{s}{n}}\).
                    \item[(2)] Use $\phi_1$ to pack $\{A_i,A_i,\dots,A_i\}$ for $1\leq i\leq n$ and $\{B_1,B_2,\dots,B_n\}$ into $\mathcal{A}_i\in\mathsf{GR}_{\sqrt{m}}^{\frac{t}{n}\times r}$ for $1\leq i\leq n$ and $\mathcal{B}\in\mathsf{GR}_{\sqrt{m}}^{r\times \frac{s}{n}}$, respectively; 
                    \item[(3)] Call Batch-$\EP$ with inputs \(\{\mathcal{A}_1, \mathcal{A}_2, \dots, \mathcal{A}_n\}\) and \(\{\mathcal{B}, \mathcal{B}, \dots, \mathcal{B}\}\), producing the output \(\{\mathcal{A}_1\mathcal{B}, \dots, \mathcal{A}_n\mathcal{B}\}\). More specifically, use \(\phi_2\) to pack \(\{\mathcal{A}_1, \mathcal{A}_2, \dots, \mathcal{A}_n\}\) and \(\{\mathcal{B}, \mathcal{B}, \dots, \mathcal{B}\}\) into \(\mathbb{A} \in \mathsf{GR}_m^{\frac{t}{n} \times r}\) and \(\mathbb{B} \in \mathsf{GR}_m^{r \times \frac{s}{n}}\), respectively. Next, call EP codes over \(\mathsf{GR}_m\) to compute \(\mathbb{C} = \mathbb{A}\mathbb{B}\). Finally, use \(\psi_2\) to unpack \(\mathbb{C}\) to obtain \(\{\mathcal{A}_1\mathcal{B}, \dots, \mathcal{A}_n\mathcal{B}\}\).

                    \item[(4)] For $1\leq i\leq n$, use $\psi_1$ to unpack $\mathcal{A}_i\mathcal{B}$ to get $\{A_iB_1,A_iB_2,\dots,A_iB_n\}$. Then output 
                    \[C=\begin{pmatrix}
                                        A_1B_1 & A_1B_2 &\dots & A_1B_n \\
                                        \vdots \\
                                        A_nB_1 & A_nB_2 & \dots & A_nB_n
                                    \end{pmatrix}\]
                \end{itemize}
            \end{minipage}
        }

        \medskip

        \begin{Cor}\label{cor:poly}
            Let $A, B, N, m, n$ be defined as in $\EP$-I. Assume $u\mid \frac{t}{n},v\mid \frac{s}{n},w\mid r$ are divisors such that $R:=uvw+w-1\leq N$. The recovery threshold in the above $\EP$-II is $R$.  Moreover, the communication complexities are
            \begin{itemize} 
                \item[--] Upload: $O\left((\frac{trv+sru}{uvw} )\sqrt{m}N\right)$.
                \item[--] Download: $O\left(\frac{ts}{uv}R\right)$. 
            \end{itemize}
            The computational costs are
            \begin{itemize}
                \item[--] Encoding: $\tilde{O}\left((\frac{trv+sru}{uvw})\sqrt{m}N\log^2 N \right)$.
                \item[--] Decoding: $\tilde{O}(\frac{ts}{uv}R\log^2 R)$.
                \item[--] The computation of each worker node: $\tilde{O}\left(\frac{tsr}{uvw}\right)$.
            \end{itemize}
        \end{Cor}
        \begin{proof}
            The proof is largely similar to that of Corollary~\ref{cor:mat}, with the primary difference being the use of RMFE twice for packing and unpacking, where \( m = O(n^2) \).
        
            The upload and download costs come from Step (3). By Lemma~\ref{Lem:3.1} with $m=O(n^2)$, and by substituting $t$ and $s$ with $t/n$ and $s/n$, respectively, the upload cost is $O\left((\frac{trv+sru}{nuvw} )mN\right)=O\left((\frac{trv+sru}{uvw} )\sqrt{m}N\right).$ Similarly, the download cost is  $O\left(\frac{ts}{n^2uv}mR\right)=O\left(\frac{ts}{uv}R\right)$. 
        
            The computation of each Worker node is $\tilde{O}\left(\frac{tsr}{n^2uvw}\tilde{O}(m)\right)=\tilde{O}\left(\frac{tsr}{uvw}\right)$.
            
            For the encoding and decoding complexity of $\EP$-II, it suffices to show that the encoding and decoding complexity of plain EP codes by Lemma~\ref{Lem:3.1} with $m=O(n^2)$ dominate the complexity of RMFEs, respectively.
        
            The encoding process consists of Steps (2) and (3). The computation of RMFE maps $\phi_1$ in Step (2) requires $\left(n(\frac{t}{n}r)+(\frac{s}{n}r)\right)\tilde{O}(n)\leq \tilde{O}((tr+sr)\sqrt{m})$ operations in $\mathsf{GR}$. The computation of RMFE maps $\phi_2$ in Step (3) requires $\left((\frac{t}{n}r)+(\frac{s}{n}r)\right)\tilde{O}(n\cdot n)\ \leq \tilde{O}((tr+sr)\sqrt{m})$ operations in $\mathsf{GR}$. Thus, the total complexity of Steps (2) and (3) is dominated by $\tilde{O}\left((\frac{trv+sru}{uvw})\sqrt{m}N\log^2 N \right)$ as $uvw\leq R\leq N$. 
            
            The decoding process consists of Steps (3) and (4). By similar analysis as in the encoding complexity, the computation of RMFE maps $\psi_1$ in Step (3) and $\psi_2$ in Step (4) only require $\tilde{O}(st)$ operations, which is dominated by the decoding complexity of plane EP codes over $\mathsf{GR}_m$, i.e., $\tilde{O}(\frac{ts}{n^2uv}mR\log^2 R)=\tilde{O}(\frac{ts}{uv}R\log^2 R)$.
        \end{proof}

        \begin{Rmk}[Comparisons of Single CDMM based on EP codes]
            Recall that $m=\lceil(\log_p N)/d\rceil$. By Corollaries~\ref{cor:mat} and ~\ref{cor:poly}, we see that $\EP$-I (resp. $\EP$-II) saves the complexities of encoding (resp. decoding), upload, and the computation of each Worker node by a factor of $m$ than the plain EP codes over $\mathsf{GR}_m$ presented in Lemma~\ref{Lem:3.1}. The download and decoding complexity of $\EP$-I remains the same as plain EP codes. However, the upload and encoding complexities of $\EP$-II are still smaller than plain EP codes by a factor of $1/\sqrt{m}$.
        \end{Rmk}

        \begin{Rmk}[Comparison of $\EP$ with CDMM from Algebraic geometry codes]\label{rmk:AGcomp}
            For a small Galois field $\mathsf{GR}(p,d)=\mathsf{GF}({p^d})$ with $p^d<N$, there exists algebraic function field $F/\mathsf{GF}({p^d})$ with at least $N$ rational places such that the CDMM with distributed order $N$ can be extended to algebraic codes defined on $F/\mathsf{GF}(q)$. By the state-of-art AG-based Polydot CDMM \cite{DBLP:journals/corr/abs-2408-01806},  given matrices ${A} \in \FF_q^{t \times r}$ and ${B} \in \FF_q^{r \times s}$, the complexities for AG-based codes\cite{DBLP:journals/corr/abs-2408-01806} to compute ${C} = {A}{B}$ is as follows:
            \begin{itemize}
                \item[-] Encoding: ${O}\left((\frac{trv+sru}{uvw})N^3 \right)$;
                \item[-] Decoding: $O\left(\frac{ts}{uv} R^2 + R^3 \right)$;
            \end{itemize}
            where \( R \approx  (2w+1)uv+4g \) is recovery threshold and $g$ is the genus of $F/\mathsf{GF}(q)$. Other metrics are
            \begin{itemize}
                \item[-] Upload: $O\left(\left(\frac{tr}{uw} + \frac{rs}{wv}\right) N\right)$;
                \item[-] Download: $O\left(\frac{ts}{uv} R \right)$;
                \item[-] The computaion of each worker node: $O\left(\frac{trs}{uwv} \right)$.
            \end{itemize}
            By Corollaries~\ref{cor:mat} and \ref{cor:poly}, we see that the encoding (resp. decoding) complexity of AG-based CDMM is larger than our $\EP$-I and $\EP$-II by a multiple of $N$ (resp. $R$), which is rather expensive for higher distributed order $N$ and recovery threshold $R$. However, our single $\EP$-I and $\EP$-II cannot simultaneously achieve the optimal upload and download complexities as AG-based CDMM does. Specifically, $\EP$-I achieves optimal upload but a sacrifice of worsen download by a multiple of $m$ than AG-based DMM; while $\EP$-II achieves optimal download but a sacrifice of worsen upload by a multiple of $\sqrt{m}$ than AG-based DMM;
        \end{Rmk}

    \section{Evaluations of Single $\EP$}\label{sec:4}
        In this section, we implement and evaluate the performance of our method for computing single large matrix multiplication over \(\ZZ_{2^e}\). We analyze the improvements in computation cost and communication volume for both the master and worker nodes, comparing our approach with the original DMM using EP codes.

        \subsection{Setup}
            We implemented both the original DMM (EP) and our two optimized versions (i.e., \(\EP\)-I and \(\EP\)-II) using RMFE in \verb!C++!, based on the \verb!NTL! library\footnote{\href{https://libntl.org/}{https://libntl.org/}}.
            
            Our experiments were conducted on a Supermicro SYS-7049GP-TRT server, equipped with two Intel(R) Xeon(R) Gold 5220R CPUs (2.20GHz, 24 cores per processor, hyper-threading) and 128GB of RAM. Each node operated within a \verb!Docker! container, restricted to single-thread execution. We conducted each experiment 10 times per configuration and averaged the results for accuracy.

            In our evaluation, the matrices were defined over \(\ZZ_{2^{64}}\) (i.e., \(\mathsf{GR}(2^{64},1)\)), and we evaluated two scenarios with varying numbers of worker nodes, which required different extension degrees for computation. The first setup, with 8 worker nodes, corresponded to computations over \(\mathsf{GR}(2^{64}, 3)\), while the second setup, with 16 worker nodes, used \(\mathsf{GR}(2^{64}, 4)\). For DMM partitioning, we set parameters \(u = v = 2\) and \(w = 1\) for the 8-worker case, resulting in a recovery threshold of 4. In the 16-worker case, we set \(u = w = v = 2\), leading to a threshold of 9. Both optimization types were configured with \(n = 2\). Since we only tested small Galois rings with \(m = 3\) or \(m = 4\), we did not split matrix \(A\) in \(\EP\)-II and applied only \(\phi_1\).

            We tested matrix multiplications of square matrices, varying the size to assess scalability and performance impact across dimensions. Matrix sizes included 2000 (for two \(2000\times2000\) matrices), 4000, 6000, and 8000, with each dimension representing the height (or width) of the square matrices involved.

            It is important to note that our Galois ring is implemented using the \verb!NTL! library directly, without optimization, and all implementations are single-threaded, omitting other potential optimization techniques. Consequently, the computation time is less competitive compared to highly optimized field-based versions of DMM. This approach, however, allows for a direct comparison with the original DMM (EP codes) method, isolating the impact of our optimizations without the influence of additional performance enhancements. As interest in Galois rings grows and more research focuses on their applications, we anticipate the emergence of more efficient implementations in the near future. At that point, integrating the techniques proposed in this work could lead to significantly improved performance.
    
        \subsection{Master Node}
            \begin{figure}
                \centering
                \subfloat[Computation]{\label{fig:ste3}\includegraphics[width=0.5\linewidth]{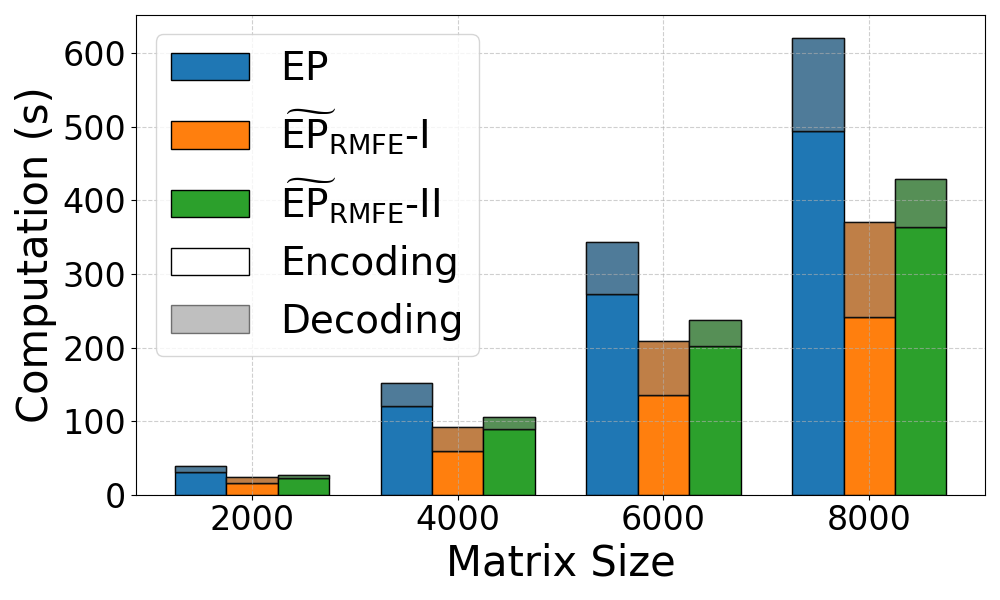}}
                \subfloat[Communication]{\label{fig:sce3}\includegraphics[width=0.5\linewidth]{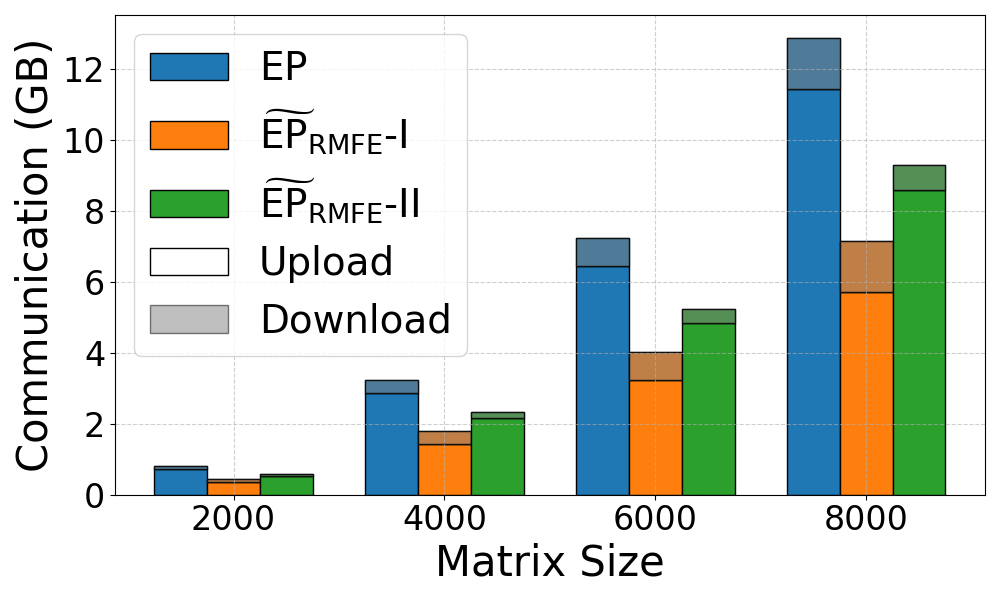}}
                \caption{Comparison of the computation time and communication volume of master node in the case of 8 worker nodes.}
                \label{fig:se3}
            \end{figure}

            \begin{figure}
                \centering
                \subfloat[Computation]{\label{fig:ste4}\includegraphics[width=0.5\linewidth]{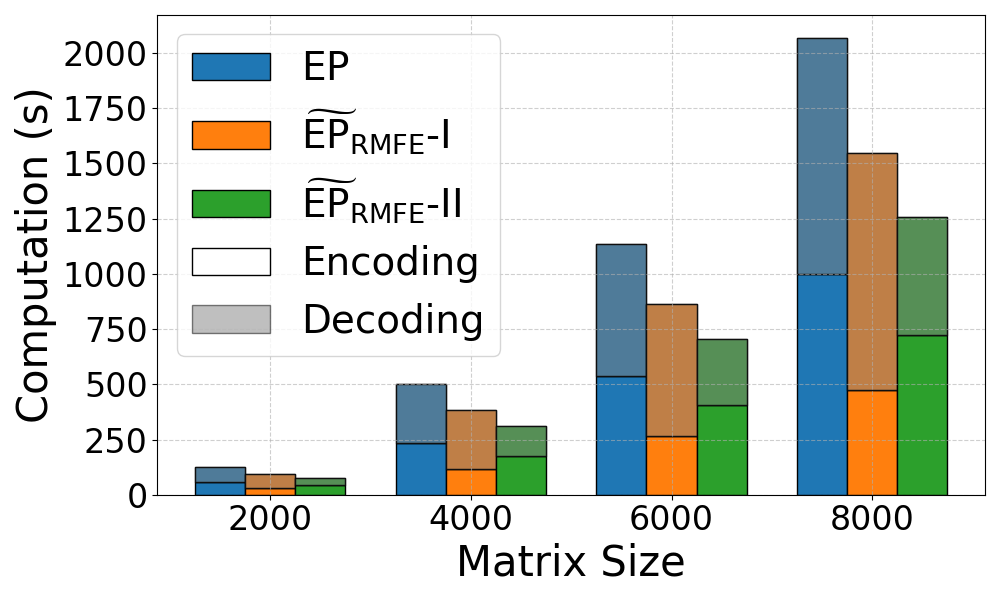}}
                \subfloat[Communication]{\label{fig:sce4}\includegraphics[width=0.5\linewidth]{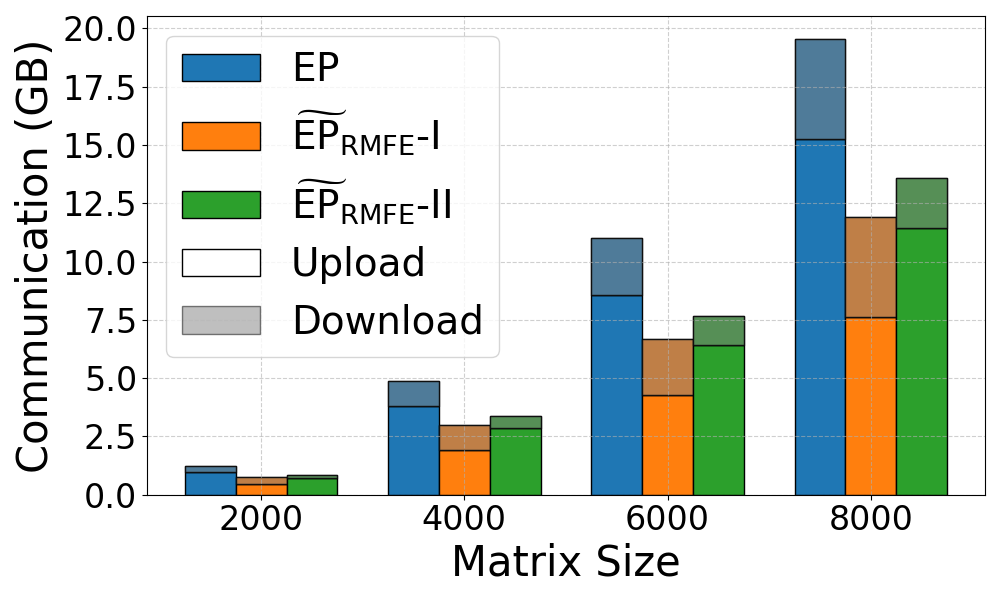}}
                \caption{Comparison of the computation time and communication volume of master node in the case of 16 worker nodes.}
                \label{fig:se4}
            \end{figure}
            The evaluation results for the master node's computation time and communication volume shown in Figure~\ref{fig:se3} and Figure~\ref{fig:se4}, reflect consistent trends across all scenarios and matrix sizes. The baseline EP consistently shows the highest computation time and communication volume. In contrast, \(\EP\)-I reduces encoding time by half while maintaining the same decoding time. For communication, \(\EP\)-I I achieves a 50\% reduction in upload volume, while download volume remains unaffected. These efficiencies make \(\EP\)-I particularly well-suited for configurations where decoding time represents a minor portion of the overall computation time, as seen in the 8-worker node scenario over \(\mathsf{GR}(2^{64},3)\).

            On the other hand, \(\EP\)-II offers a balanced performance between EP and \(\EP\)-I for encoding time, while significantly reducing decoding time by half. This is advantageous in scenarios where decoding time is a substantial component of the total computation time. For communication, \(\EP\)-II reduces download volume by half, with upload volume positioned between EP and \(\EP\)-I. This configuration proves particularly advantageous in the 16-worker node setup over the larger Galois ring \(\mathsf{GR}(2^{64},4)\), where decoding time becomes more prominent in the computation workload. Thus, \(\EP\)-II outperforms in this setup, with a potential for further gains as the number of worker nodes and Galois ring size increase.

        \subsection{Worker Node}
            \begin{figure}
                \centering
                \subfloat[Computation]{\label{fig:cte3}\includegraphics[width=0.5\linewidth]{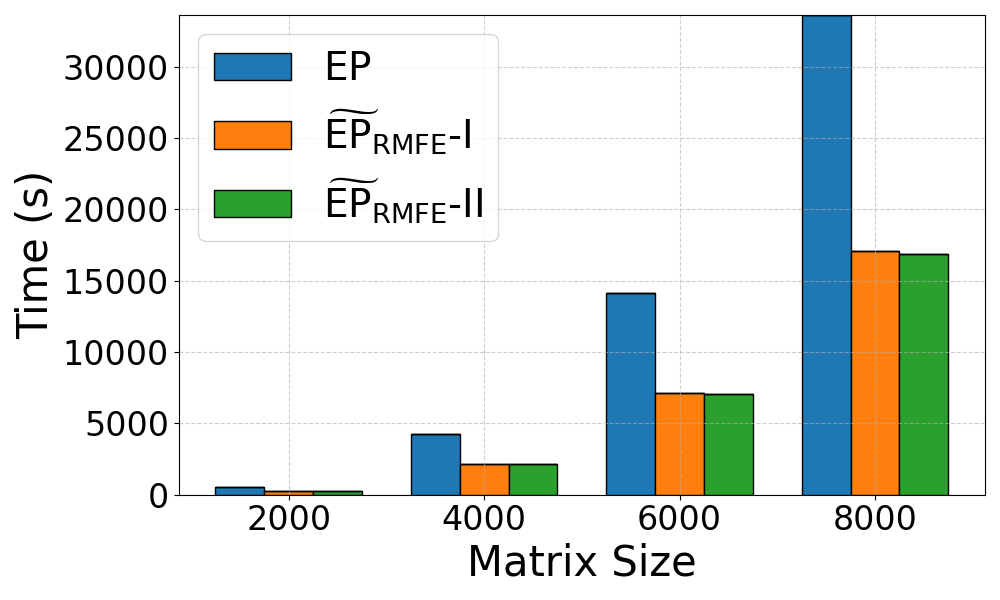}}
                \subfloat[Communication]{\label{fig:cce3}\includegraphics[width=0.5\linewidth]{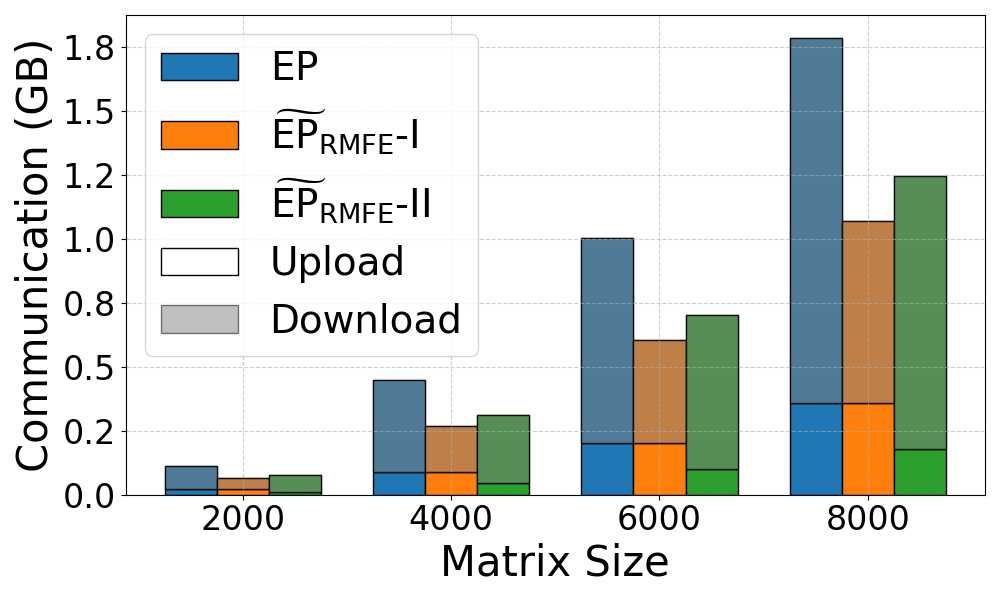}}
                \caption{Comparison of the computation time and communication volume of worker node in the case of 8 worker nodes.}
                \label{fig:ce3}
            \end{figure}

            \begin{figure}
                \centering
                \subfloat[Computation]{\label{fig:cte4}\includegraphics[width=0.5\linewidth]{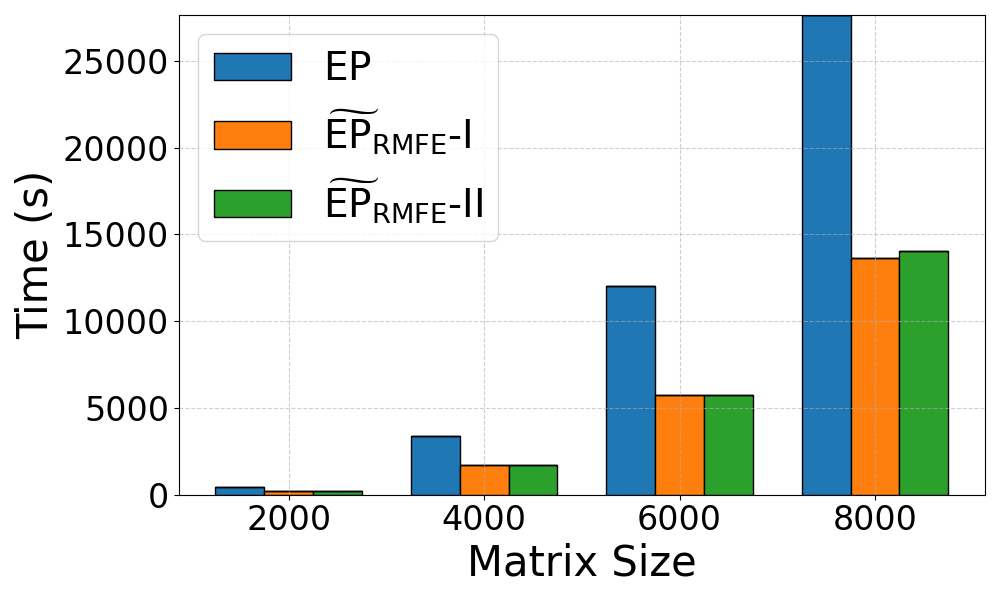}}
                \subfloat[Communication]{\label{fig:cce4}\includegraphics[width=0.5\linewidth]{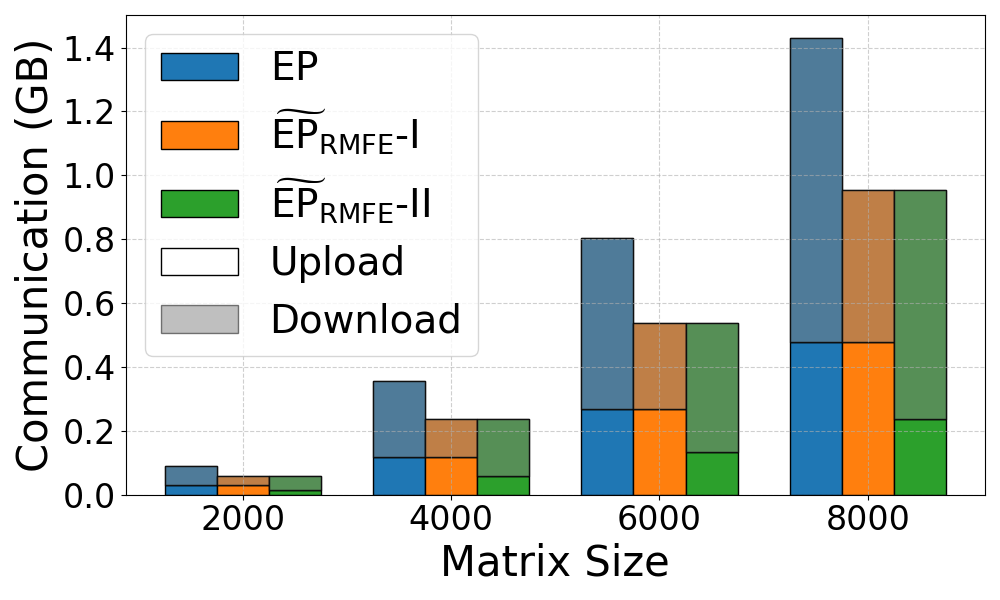}}
                \caption{Comparison of the computation time and communication volume of worker node in the case of 16 worker nodes.}
                \label{fig:ce4}
            \end{figure}
            The evaluation results for computation time and communication volume per worker node are presented in Figure~\ref{fig:ce3} and Figure~\ref{fig:ce4}. Both \(\EP\)-I and \(\EP\)-II reduce computation time by half compared to the baseline in both the 8- and 16-worker configurations. For communication, the figures illustrate results only for the worker nodes involved in the recovery process, as nodes not participating in recovery have an upload volume of zero. Communication patterns on worker nodes mirror those observed on the master node: reductions in the master’s upload volume translate to reductions in the workers' download volume, and vice versa.

            When comparing Figure~\ref{fig:cte3} (8-worker case) and Figure~\ref{fig:cte4} (16-worker case), we observe that with the same matrix size, computation costs are lower with more worker nodes, despite the larger Galois ring (and its associated higher computation cost). This is because additional worker nodes allow for finer partitioning of the original matrix, resulting in smaller sub-matrices per worker node. The benefit of this reduction in per-node workload outweighs the increased cost of operations over a larger Galois ring.

            Furthermore, with a larger Galois ring, RMFE can be leveraged to pack more elements together, further reducing both computation and communication costs. For example, in a configuration with 32 worker nodes (requiring operations over \(\mathsf{GR}(2^{64}, 5)\), setting \(n=3\) and using a \((3,5)\)-RMFE enables a more efficient packing strategy, thus optimizing the workload for each worker node even further.

    \section{Conclusion}
        In this paper, we construct efficient CDMM schemes over general Galois rings, introducing a versatile CDMM framework for batch matrix multiplications via the RMFE technique. Building on this, we further optimize single matrix multiplication by applying different batch processing strategies to input matrices under our CDMM framework, resulting in two new Single CDMM schemes, \(\EP\)-I and \(\EP\)-II. Evaluations substantiate the effectiveness of our methods, demonstrating substantial reductions in computation time and communication volume compared to the original EP codes. Specifically, \(\EP\)-I minimizes encoding and upload costs, making it well-suited for cases with limited network bandwidth, while \(\EP\)-II reduces decoding and download costs, enhancing performance in computation-heavy settings. Our experiments highlight these optimizations' potential to scale efficiently across a wide range of distributed computing environments, reinforcing the feasibility of CDMM applications over Galois rings and laying the groundwork for future improvements in fault-tolerant, high-performance matrix multiplication.

    \bibliographystyle{IEEEtran}
    \bibliography{reference}

\end{document}